\newcommand{\version}{September 1, 2011}

\documentclass[12pt]{article}
\usepackage{bbm}
\usepackage{a4,amsthm,amsfonts,latexsym,amssymb}
\swapnumbers 

\theoremstyle{plain}
\newtheorem{thm}{THEOREM}

\newcommand{\beq}{\begin{equation}}
\newcommand{\eeq}{\end{equation}}
\def\beqa{\begin{eqnarray}}
\def\eeqa{\end{eqnarray}}

\newcommand{\Tr}{{\rm Tr}}

\newcommand{\Hh}{{\mathcal H}}

\newcommand{\K}{{\mathcal K}}

\newcommand{\tdt}{\!\cdot\!}

\def\bce{\begin{center}}
\def\ece{\end{center}}
\def\bit{\begin{itemize}}
\def\eit{\end{itemize}}

\date{\small\version}
\begin{document}
\markboth{\scriptsize{An Inequality for the trace of matrix products, using absolute values \qquad  \version}}
         {\scriptsize{An Inequality for the trace of matrix products, using absolute values \qquad  \version}}

\title{
%%\vspace{-80pt}
%%%\begin{flushright}
%%%{\small Vienna, ESI Report xxxx (2011)} \vspace{30pt}
%%%\end{flushright}
%%\begin{center}
\bf{An Inequality for the trace of matrix products, using absolute values}}
%%\end{center}
\vspace{60pt}
\author{ \vspace{20pt} Bernhard Baumgartner$^1$
\\
\vspace{-4pt}
\small{Fakult\"at f\"ur Physik, Universit\"at Wien}\\
\small{Boltzmanngasse 5, A-1090 Vienna, Austria}}

\maketitle

%%% \textsl{\bf Draft}\qquad \version

\vspace{60pt}

\begin{abstract}
The absolute value of matrices is used in order to give inequalities for the trace of products.
An application gives a very short proof of the tracial matrix H\"{o}lder inequality.
\\[20ex]
\\[3ex]

\end{abstract}

\footnotetext[1]{\texttt{Bernhard.Baumgartner@univie.ac.at}}

%%\maketitle

%%%%%%%%%%%%%%%%%%%%%%%%%%%%%%%%%%%%%%%%%%%%%%%%%%%%%%%%%%%%%%%%%%%%%%%%%%%%%%%%
\newpage

\section{The baby inequality and its application}\label{baby}

\begin{thm} \textbf{Using absolute values:}\label{absolute}                        %%baby
Consider two complex $m\times n$ matrices $A$, $B$ and their absolute values,
$|A|=(A^\ast A)^{1/2},\quad |A^\ast|=(A A^\ast)^{1/2}$.
Then
\beq \label{babyformula}
|\Tr\, A^\ast B| \leq (\Tr\, |A|\tdt |B|)^{1/2} \cdot (\Tr\, |A^\ast|\tdt |B^\ast|)^{1/2}
\eeq
\end{thm}

\begin{proof}
Let $e_i$ be an orthonormal basis made of normalized eigenvectors of $A^\ast A$, with eigenvalues $a_i^2 \geq0$.
For $a_i\neq 0$,  $f_i=a_i^{-1}Ae_i$ are normalized eigenvectors of $AA^\ast$, obeying
$A^\ast f_i=a_ie_i$. Eventually, to make a full basis, this set has to be completed
by introducing eigenvectors of  $AA^\ast$ with eigenvalue $0$.
Analogously, there are basis-sets of normalized vectors $g_j$ and $h_j$,
such that $Bg_j=b_j h_j$, $B^\ast h_j=b_j g_j$, with $b_j\geq 0$.
With these vectors we get
\beqa
|\Tr\, A^\ast B|   &=& \left| \sum_{i,j} a_i b_j \langle g_j, e_i\rangle\langle f_i, h_j\rangle\right|.\label{babyproof1} \\
                    &&         \textrm{Applying the Cauchy-Schwarz inequality gives}  \nonumber \label{babyproof2}\\
 |\Tr\, A^\ast B|  &\leq & \left( \sum_{i,j} a_i b_j \langle e_i,g_j\rangle\langle g_j,e_i\rangle \right)^{1/2}
                     \left( \sum_{i,j} a_i b_j \langle f_i,h_j\rangle\langle h_j,f_i\rangle \right)^{1/2} \\
             &=& \left(\Tr\, |A|\tdt|B| \right)^{1/2} \left(\Tr\, |A^\ast|\tdt|B^\ast| \right)^{1/2}
\eeqa
In the last step the identities $|A|e_i=a_ie_i$,
 $|A^\ast|f_i=a_if_i$,  $|B|g_j=b_jg_j$, and   $|B^\ast|h_j=b_jh_j$
 have been used.
\end{proof}

We remark that the inequality is sharp.
It becomes an equality in case both matrices $A$ and $B$ have rank one.
This follows from the fact that the ``sum'' in (\ref{babyproof1}) consists of only one term,
so the Schwarz inequality in (\ref{babyproof2}) becomes an equality.
A simple example is given with $2\times 2$ matrices:
\beqa \label{twomatrices}
A&=&\left( \begin{array}{cc}1&0\\ 0&0 \end{array}\right), \quad
B=\left( \begin{array}{cc}1&1\\ 0&0 \end{array}\right), \\ \nonumber
\medskip
 \textrm{with}\medskip\quad &|B|&=\quad\frac 1 {\sqrt{2}}\left( \begin{array}{cc}1&1\\ 1&1 \end{array}\right),\qquad
|B^\ast|\quad=\quad\sqrt 2 \left( \begin{array}{cc}1&0\\ 0&0 \end{array}\right);
\eeqa
\medskip
so \quad $\Tr\,A^\ast B=1$, \quad $\Tr\,|A|\tdt|B|=1/\sqrt 2$, \quad $\Tr\,|A^\ast|\tdt|B^\ast|=\sqrt 2$.
\newline
Other cases where (\ref{babyproof2}) becomes an equality are $\langle e_i,g_j\rangle = \alpha \langle f_i,h_j\rangle$
$\forall i, j$, for some constant $\alpha$.

The example shows also that the distinction between the absolute values $|B|$ and $|B^\ast|$
is necessary. Only
if both $A$ and $B$ are normal matrices the absolute values are equal, $|A^\ast|=|A|$, $|B^\ast|=|B|$.
In such a case equation (\ref{babyformula}) becomes
$|\Tr\, A^\ast B| \leq \Tr\, |A|\tdt |B|$.

The following application was at the origin of my investigations, searching for a simple proof of the H\"{o}lder inequality
for matrices and operators. (For other proofs, see f.e. \cite{MBR72, RS75, RB97, EC09}\,)

\begin{thm} \textbf{Matrix H\"{o}lder Inequality:}
Consider two \,$m\times m$\, matrices $A$, $B$ and their absolute values, then
\beq
|\Tr\, A^\ast B|\leq \left(\Tr\,|A|^p \right)^{1/p}\cdot \left(\Tr\,|B|^q \right)^{1/q},
\qquad 1\leq p,q\leq\infty,\quad p^{-1}+q^{-1}=1
\eeq
\end{thm}

\begin{proof}
Using the same notation as in the proof of Theorem \ref{absolute}, we note that $\Tr\, |A|^p=\Tr\, |A^\ast|^p=\sum_i a_i^p$.
So, the H\"{o}lder Inequality for the left hand side of (\ref{babyformula}) is proven, if it holds for each factor on the right hand side.
There we have normal operators. For these we can use the classical H\"{o}lder Inequality for weighted sums,
followed by using completeness of the basis sets $\{e_i\}$ and $\{g_j\}$:
\beqa
\Tr\, |A|\tdt |B|= \sum_{i,j} a_i b_j \, |\langle e_i,g_j\rangle|^2 &\leq & \nonumber
\left( \sum_{i,j} a_i^p \, |\langle e_i,g_j\rangle|^2\right)^{1/p}\cdot
\left( \sum_{i,j} b_j^q \, |\langle e_i,g_j\rangle|^2\right)^{1/q}\\
=\left( \sum_i a_i^p \right)^{1/p}\cdot \left( \sum_j b_j^q \right)^{1/q} &=& \nonumber
(\Tr\,|A|^p)^{1/p} \cdot (\Tr\,|B|^q)^{1/q}.
\eeqa
Analogously for $\Tr\, |A^\ast|\tdt |B^\ast|$.
\end{proof}

If one is interested in characterizing cases of equality for the matrix H\"{o}lder inequality,
one has to check whether  $\langle e_i,g_j\rangle = \alpha \langle f_i,h_j\rangle$,
as stated above, and also whether the classical  H\"{o}lder inequality used in the proof becomes an equality.

\section{Generalizations}\label{general}
There are possibilities to generalize the baby inequality. It can grow by:
Inserting extra matrices, one $m\times m$ another one $n\times n$;
using different exponents for the different absolute values;
going into vector spaces with infinite dimension.

One can insert extra matrices $M$ and $N$:
\begin{thm}\textbf{Inequality with intermediate matrices:}
\beq
|\Tr\, M A^\ast N B|\leq \left(\Tr\,M \,|A|\, M^\ast\, |B| \, \right)^{1/2}\cdot \left(\Tr\,N^\ast \,|A^\ast|\, N\, |B^\ast|\, \right)^{1/2}
\eeq
\end{thm}

\begin{proof}
Just insert the extra matrices in the right places in the inner products appearing in (\ref{babyproof1}) and (\ref{babyproof2}):
Replace $\langle g_j,e_i\rangle$ by $\langle g_j,M e_i\rangle$ and $\langle f_i,h_j\rangle$ by  $\langle f_i,N h_j\rangle$.
\end{proof}

There is the possibility to consider different exponents:
\begin{thm}\textbf{Inequality with exponents:}
\beq
|\Tr\, A^\ast B| \leq \left(\Tr\,|A|^\alpha \, |B|^\beta \, \right)^{1/2}\cdot \left(\Tr\,|A^\ast|^{2-\alpha}\, |B^\ast|^{2-\beta} \right)^{1/2}, \qquad  0\leq \alpha, \beta \leq 2, \nonumber
\eeq
 with $0^0=0$, so that $|A|^0= \textrm{projector onto range}(|A|)$.
\end{thm}

\begin{proof}
Modify (\ref{babyproof1}) and (\ref{babyproof2}) as
\beq
|\Tr\, A^\ast B|   = \left| \sum_{i,j} \varphi_{i,j}\cdot\psi_{i,j} \right|
                   \leq \left( \sum_{i,j} |\varphi_{i,j}|^2\right)^{1/2} \cdot \left( \sum_{i,j} |\psi_{i,j}|^2\right)^{1/2},
\eeq
with \quad $\varphi_{i,j}= a_i ^{\alpha/2} b_j ^{\beta/2} \langle e_i,g_j\rangle$,
\quad $\psi_{i,j}= a_i ^{1-\alpha/2} b_j ^{1-\beta/2} \langle h_j,f_i\rangle$.
Observe, that for $a_i=0$ the matrix elements involving $e_i$ or $f_i$ are just absent.
The same holds for $b_j, g_j, h_j$.
\end{proof}

Extensions into infinite dimensions can be done in different ways.
I present the following result:
\begin{thm}\textbf{Inequality for two Hilbert Schmidt class operators:}
Let $A$ and $B$ be operators from Hilbert space $\Hh$ to the Hilbert space $\K$, with the
properties $\Tr_\Hh \, A^\ast A=\Tr_\K \, A\tdt A^\ast<\infty$ and  $\Tr_\Hh \, B^\ast B=\Tr_\K \, B\tdt B^\ast<\infty$. Then
\beq
|\Tr_\Hh \,A^\ast B|\leq \left(\Tr_\Hh \,|A|\,|B|\,\right)^{1/2}\cdot \left(\Tr_\K \,|A^\ast|\,|B^\ast| \right)^{1/2} \label{HSinequ}
\eeq
\end{thm}

\begin{proof}
As in the proof of Theorem \ref{absolute} we use the singular values $a_i$ and basis sets $e_i$, $f_i$,
here $e_i\in\Hh$, $f_i\in\K$,        such that $Ae_i=a_i f_i$, and analogously  $B g_i=b_i h_i$.
In Dirac's notation $A=\sum_i|f_i\rangle a_i\langle e_i|$ and $B=\sum_i|h_i\rangle b_i\langle g_i|$.
$A$  being in the Hilbert-Schmidt class means
$$\Tr_\Hh \, A^\ast A= \Tr_\K \, A\tdt A^\ast =\Tr_\Hh \, |A|^2=\Tr_\K \,|A^\ast|^2=\left(\sum_i a_i^2\right)^{1/2}<\infty ,$$
and the analogue for $B$.
Introducing the operators with finite rank
\beq
 A_N=\sum_i^N|f_i\rangle a_i\langle e_i|  \qquad \textrm{and} \qquad B_N=\sum_i^N|h_i\rangle b_i\langle g_i|,
\eeq
one can apply  Theorem \ref{baby} to the matrices which represent these operators,
giving the inequality (\ref{HSinequ}) for $A_N$ and $B_N$.
One observes the convergences in norm:
$$\|A-A_N\|=\|A^\ast-A^\ast_N\|=\|\,|A|-|A_N|\,\|=\|\,|A^\ast|-|A^\ast_N|\,\|=
\left(\sum_{N+1}^\infty a_i^2\right)^{1/2} \rightarrow_{\small N\rightarrow\infty} 0,$$
and the same for $B$.
The Hilbert-Schmidt inner products are jointly norm-continuous in both factors,
so each side of the inequality (\ref{HSinequ}) for $A_N$ and $B_N$ converges as $N\rightarrow\infty$,
giving the same inequality without the $N$ as an index.
\end{proof}

Applications are new proofs for H\"{o}lder type inequalities used in mathematical physics.
They will be discussed in a following article.

\section{Comparison with another use of absolute values}\label{otheruse}

The product $A^\ast\tdt B$ can be represented as
\beq\label{product}
A^\ast\tdt B = U\tdt|A^\ast|\tdt|B^\ast|\tdt V^\ast,
\eeq
by extending the $m\times n$ matrices to $N\times N$ matrices, where $N=\max\{m,n\}$,
and using the polar decompositions $A^\ast = U\tdt |A^\ast|$,  $B^\ast = V\tdt |B^\ast|$
with unitary operators $U$, $V$.
(Equivalently, one can stay with the $m\times n$ matrices and use isometries $U$ and $V$
instead of unitary operators.)
This equality implies that the set of singular values of $A^\ast\tdt B$ is identical to
that of $|A^\ast|\tdt|B^\ast|$. (Eventually, when staying with $m\neq n$, the numbers of zeroes are different.)
So, all the unitarily invariant norms, see \cite{RB97}, are identical.
One identity is for the operator norm
\beq
\| A^\ast\tdt B \| = \| \,|A^\ast|\tdt|B^\ast|\,\|,
\eeq
another one gives
\beq
\Tr\, |A^\ast\tdt B| = \Tr\,|\,(|A^\ast|\tdt|B^\ast|)\, |.
\eeq
Together with $|\Tr M|\leq \Tr \, |M|$, which holds for each matrix, we get
\beq \label{weakeruse}
\left| \Tr\, A^\ast\tdt B \right| \leq \Tr\,|\,(|A^\ast|\tdt|B^\ast|)\, |.
\eeq
This inequality is not as sharp as the baby inequality given in Theorem \ref{absolute}.
As an example use the same matrices as in equation (\ref{twomatrices}).
They give $1$ on the l.h.s. but $\sqrt{2}$ on the r.h.s. of (\ref{weakeruse}).

\end{document}